\documentclass[conference,letterpaper]{IEEEtran}

\usepackage[utf8]{inputenc} 
\usepackage[T1]{fontenc}
\usepackage{url}
\usepackage{ifthen}
\usepackage{cite}
\usepackage[cmex10]{amsmath}
\usepackage{centernot}
\usepackage{amsthm}
\usepackage{enumitem}

\usepackage{graphicx}
\usepackage{graphics}

\newtheorem{theorem}{Theorem}[section]
\newtheorem{lemma}[theorem]{Lemma}

\newtheorem{example}{Example}[section]

\renewcommand{\baselinestretch}{0.96}

\begin{document}
\title{Independence Properties  of Generalized Submodular Information Measures} 


\author{%
  \IEEEauthorblockN{Himanshu Asnani}
  \IEEEauthorblockA{School of Technology and Computer Science\\
                    TIFR, Mumbai\\ 
                    Email: himanshu.asnani@tifr.res.in}
                    
  \and
  \IEEEauthorblockN{Jeff Bilmes}
  \IEEEauthorblockA{Department of ECE\\
                    University of Washington, Seattle\\ 
                    Email: bilmes@uw.edu}

  \and
                      \IEEEauthorblockN{Rishabh Iyer}
  \IEEEauthorblockA{Department of Computer Science\\
                    University of Texas, Dallas\\
                    Email: rishabh.iyer@utdallas.edu}
}


\maketitle

\begin{abstract}
   Recently a class of generalized information measures was defined on sets of items parametrized by submodular functions~\cite{iyer2021submodular}. In this paper,   we propose and study various notions of independence between sets with respect to such information measures, and connections thereof. Since entropy can also be used to parametrize such measures, we derive interesting independence properties for the entropy of sets of random variables.  We also study the notion of multi-set independence and its properties.  Finally, we present optimization algorithms for obtaining a set that is independent of another given set, and also discuss the implications and applications of combinatorial independence. 
\end{abstract}

\section{Introduction}
In this paper, we consider the recently proposed class of submodular information measures~\cite{iyer2021submodular}, and study general combinatorial independence characterizations admitted by them. A set function $f: 2^V \rightarrow \mathbf{R}$ over
a finite set $V = \{1, 2, \ldots, n\}$ is \emph{submodular} 
\cite{fujishige2005submodular} if for all
subsets $S, T \subseteq V$, it holds that $f(S) + f(T) \geq f(S \cup
T) + f(S \cap T)$. Given a set $S \subseteq V$, we define the
\emph{gain} of an element $j \notin S$ in the context $S$ as $f(j | S)
= f(S \cup j) - f(S)$. A perhaps more intuitive
characterization of submodularity is as follows:
a function $f$ is submodular if it satisfies
\emph{diminishing marginal returns}, namely $f(j | S) \geq f(j | T)$
for all $S \subseteq T, j \notin T$, and is \emph{monotone} if $f(j |
S) \geq 0$ for all $j \notin S, S \subseteq V$. Submodular functions are a rich and expressive class of models which capture a number of important aspects like coverage (e.g. set cover function), representation (e.g. facility location function), diversity (e.g. log-determinants), and information (e.g. joint entropy of a set of random variables). Submodular functions have been shown to be closely connected with convexity~\cite{lovasz1983submodular,bach2011learning}, and concavity~\cite{iyer2020concave}.

Given a submodular function $f$, we can define the submodular (conditional) mutual information $I_f(A; B | C)$ as:\
\begin{align}
    I_f(A; B | C) = f(A | C) + f(B | C) - f(A \cup B | C)
\end{align}
where the conditional information $f(A | C) = f(A \cup C) - f(C)$.
Recently, \cite{iyer2021submodular,levin2020online} studied several important and interesting properties of $I_f$ such as non-negativity, monotonicity, conditions for submodularity, and upper/lower bounds. Furthermore, \cite{iyer2021submodular} also studied multi-set extensions of the submodular mutual information to capture joint information between $k$ sets. As argued by \cite{iyer2021submodular, levin2020online}, the submodular mutual information (or multi-set mutual information) effectively captures the \emph{shared} information between two sets (or multiple sets), from the lens of the submodular function $f$. \cite{iyer2021submodular} also study a number of examples of $I_f$, and the modeling capabilities of these in various machine learning applications such as query focused and privacy preserving summarization, clustering and disparate partitioning. 

In this work, we extend the paradigm of combinatorial information measures to study the concept of independence among sets. In particular, we first introduce six notions of (conditional) independence between subsets in Section II. Next, in Section III, we study the relationship between the independence types for general submodular functions, and then study these connections for a few special cases, including set cover functions, modular functions, and Entropy. We then study independence between multiple sets (Section IV), 
and then present optimization algorithms for obtaining a set $A$ which is independent of a given set $B$ (Section V), and also discuss the implications and applications of this study. While the main focus of this work is combinatorial independence, we note that our results in the special case of Entropy function (specifically Lemma 3.3), is one of the first (to our knowledge) which studies the different possible types of independence between sets of random variables. We provide proofs of our results in Section VI, and conclude this paper in Section VII.

\section{Submodular (Conditional) Independence: } 
Two random variables $X$ and $Y$ are statistically independent iff the mutual information $I(X;Y) = 0$. Since we are studying sets of variables (rather than two random variables), the notion of independence becomes a little more intricate. Consequently, we define six types of combinatorial (or submodular) independence relations between two sets $A,B \subseteq \Omega$ with respect to a submodular function $f$. These relations also hold in the entropic case, in which case it is exactly the statistical independence between two sets of random variables. Before going into defining the different independence types, we provide some definitions. 

A function $f$ is said to satisfy the condition $\mathcal{M}(S_1,S_2)$ (for any $S_1, S_2\subseteq \Omega$), if $f(\{j\}|X) = f(\{j\})$, $\forall j\in S_1\backslash X$ and $X \subseteq S_2$. This condition implies that the function $f$ is modular with respect to the addition of any element from $S_1$ to any subset of $S_2$ not containing the specific element. This notation is introduced to draw attention towards implications of the difference amongst the definitions of the independence types. Below, we define the six types of independence conditions.   

\begin{enumerate}[topsep=0pt,itemsep=0ex,partopsep=1ex,parsep=1pt,leftmargin=5mm,label=\textbf{\arabic*}.] \label{def:submod-independence-types}
\item Joint Independence ($JI$): We say set $A$ is jointly independent of another set $B$, or $A \perp_{J} B$ with respect to $f$ if $I_f(A;B) = 0$ or equivalently: $f(A|B) = f(A)$ (or $f(B|A) = f(B)$). 
\item Marginal Independence ($MI$): We say set $A$ is marginally independent with another set $B$, or $A \perp_{M} B$ with respect to $f$ if $f(\{a\}|B) = f(\{a\}) \; \forall a \in A$ and $f(\{b\}|A) = f(\{b\}) \; \forall b \in B$. In other words, $A \perp_{M} B$ if $a \perp_J B, \forall a \in A$ and $b \perp_J A, \forall b \in B$
\item Pairwise Independence ($PI$): Set $A$ is pairwise independent of another set $B$, or $A \perp_{P} B$ with respect to $f$ if $\forall a \in A, b \in B, \; f(\{a\}|\{b\}) = f(\{a\})$ and $f(\{b\}|\{a\}) = f(\{b\})$. In other words, $A \perp_{P} B$ if $a \perp_J b, \forall a \in A, b \in B$.
\item Subset Marginal Independence ($SMI$): We say set $A$ holds subset marginal independence with another set $B$, or $A \perp_{SM}B$ with respect to $f$ if
    $f(\{a\}|X) = f(\{a\})$, $\forall a \in A, X \subseteq B$ and  $f(\{b\}|X) = f(\{b\})$, $\forall b \in B, X \subseteq A$. In other words, $A \perp_{SM}B$ if $a \perp_J X, \forall a \in A, X \subseteq B$ and $b \perp_J X, \forall b \in B, X \subseteq A$. Observe that SMI generalizes MI and PI.
\item Modular Independence ($ModI$): Set $A$ holds modular independence with another set $B$, or $A \perp_{Mod}B$ with respect to $f$ if
    $f(\{j\}|X) = f(\{j\}) \; \forall j \in A \cup B \setminus X$ where 
    $X \subseteq A \cup B$. Thus, $f$ satisfies the condition $\mathcal{M}(A\cup B, A\cup B)$.
    \item Subset Modular Independence ($SModI$): Set $A$ is said to be subset modular independent of $B$ or $A \perp_{S-Mod}B$ with respect to $f$ if
    $f(\{j\}|X) = f(\{j\}) \; \forall j \in A \cup B \setminus X$ where 
    $X \subseteq A$ or $X \subseteq B$. Thus, $f$ satisfies both the conditions $\mathcal{M}(A, A\cup B)$ and $\mathcal{M}(B, A\cup B)$
\end{enumerate}
Unless specified, we will use joint independence as the default notion of independence. We also assume that $A$ and $B$ are disjoint without a loss of generality -- since joint independence implies $f(A \cap B) = 0$ and since $f$ is assumed to be monotone, this essentially means that we can remove the elements in $A \cap B$ from the ground set $\Omega$ without letting them affect the functional values.

Similar to independence, we also define conditional combinatorial independence between two sets $A$ and $C$ given a third set $B$. We say that $A, C$ are conditionally independent of each other given $B$, in the context of $f$, and denoted by $A \perp_f C \;| \; B$ iff $I_f(A;C | B) = 0$. An equivalent way of viewing this is in terms of the submodular mutual information: $A \perp_f C \;| \; B$ iff $I_f(A; B \cup C) = I_f(A; B)$ and $I_f(C; A \cup B) = I_f(C; B)$. In terms of the conditional gains, it implies that $f(C|B) = f(C | A \cup B)$ and $f(A | B) = f(A | B \cup C)$. Similar to independence, we can also define the six types of conditional independence. 

\section{Connections between Independence Types}

We start this section by providing a relationship between the different types of (conditional) independences. Figure \ref{fig:indep-relations-classes} illustrates the containment relationship between the different types of (conditional) independence. 

\begin{figure}[h]
    \centering
    \includegraphics[width=0.3\textwidth]{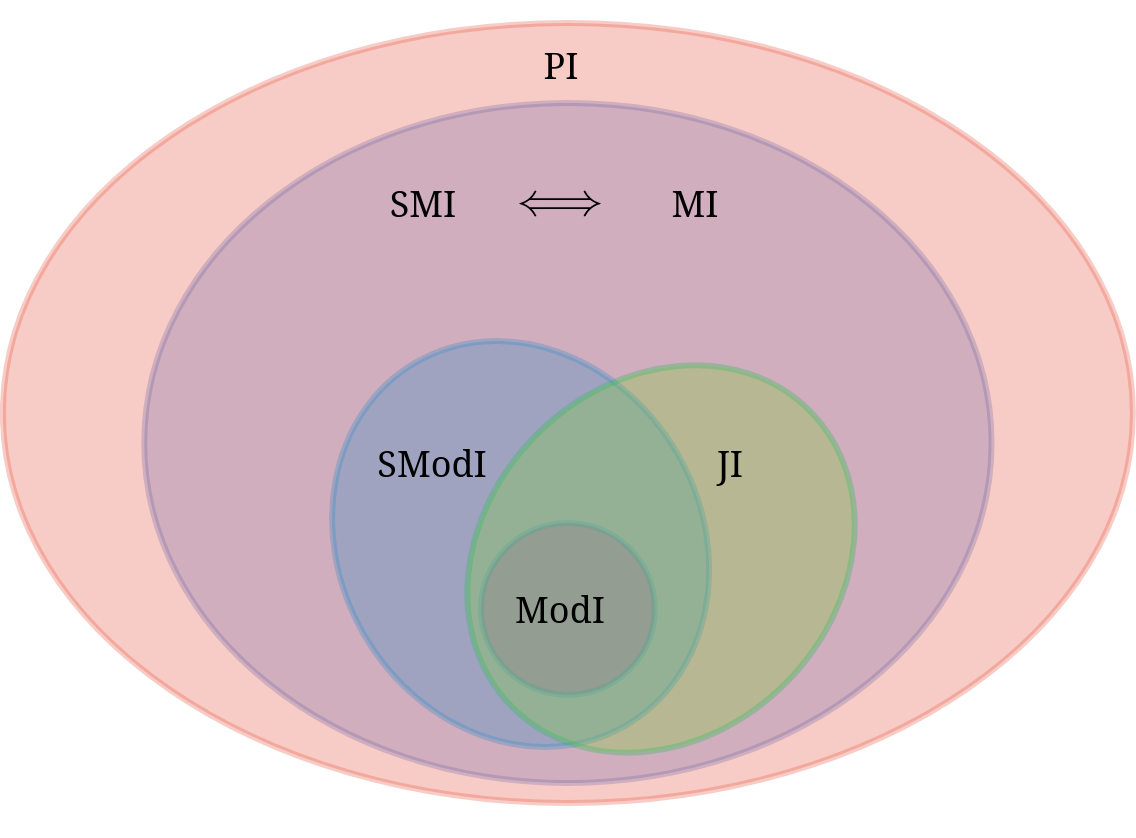}
    \caption{An illustration of the independence relation classes }
    \label{fig:indep-relations-classes}   
\end{figure}

Also, note that we prove the relationships only for the independence case and note that since conditional independence $I_f(A; B | C)$ with $f$ is equivalent to independence with $g(A) = f(A | C)$ for fixed given C, the conditions below also hold for the six types of conditional independence.

\begin{theorem}
\label{theorem:indep-relations}
The following relations hold between the different types of independence defined in the previous section. We use $JI, MI, \dots , SModI$ to denote the different types compactly. 
\begin{align}
   ModI \implies JI \implies MI \iff SMI \implies PI \; \\\text{and} \; ModI \implies SModI \implies MI
\end{align}
Moreover, there exists submodular functions where $JI \centernot\implies ModI$,  $SModI\centernot\implies ModI$,  $JI\centernot\implies SModI$, 
$SModI \centernot\implies JI$,
$MI \centernot\implies SModI$, 
$MI \centernot \implies JI$, 
$PI \centernot\implies MI$.
\end{theorem}
Thus we can see that there is a certain hierarchy attached to the independence types in general which is illustrated in Fig. \ref{fig:indep-relations-classes}.

\subsection{Relationship between the Independence Types for the Modular and Set Cover Function}
For the special case when $f$ is a modular function, all the six independence types are equivalent i.e reverse implications also hold and in fact, this follows from the very definition of a modular function. Next we show that for the Set Cover function the first four types are indeed equivalent.
\begin{lemma} \label{lemma:indep-rel-setcover}
When $f$ is a Set Cover function the relation shown in Theorem \ref{theorem:indep-relations} is not tight. In fact we have that the reverse implications also hold among the first four types, which makes them in essence equivalent to one another: $ModI \implies JI \iff MI \iff SMI \iff PI$. However, $JI \centernot\implies ModI$.
\end{lemma} 

\subsection{Relationship between the Independence Types for the Entropy Function}
As we have noted earlier, Entropy function is generalized by our framework. This is done via defining the universal set of random variables $\Omega= \{X_1, X_2, \cdots, X_k\}$ with the corresponding joint probability distribution $\mathcal{P}_{\Omega}$ on the random variables. Thus for $A\subseteq \Omega$ we have well defined submodular function, entropy $H(A)=H(X_A)$, where $X_A=\{X_a, a\in A\}$. The usual celebrated notion of independence is that of $H(A|B)=H(A) \iff P_{X_A,X_B}=P_{X_A}P_{X_B}$, or the set of random variables $X_A$ are independent of the set of random variables $X_B$. This is the $JI$ notion of independence as per our exposition on the different types of independence. However with the above definitions of six types of independence in the context of generalized submodular information, we have a richer set of notions of independence and hierarchies thereof for the entropy function. 
The lemma below further elaborates on the relations of different types of independence for the entropy function. 

\begin{lemma}\label{lemma:app-ent-ind-rel}
For the entropy function:
\begin{align}
   ModI \implies JI \implies MI \iff SMI \implies PI \; \\\text{and} \; ModI \implies SModI \implies MI
\end{align}
Moreover, there exists tuple, $(\Omega, \mathcal{P}_{\Omega})$ where $JI \centernot\implies ModI, \; SModI\centernot\implies ModI, \; JI\centernot\implies SModI, \;
SModI \centernot\implies JI, \;
MI \centernot\implies SModI, \;
MI \centernot \implies JI, \;
PI \centernot\implies MI$,
\end{lemma}

 \subsection{Miscellaneous Results}
Next, we provide sufficient but not necessary conditions for submodular (conditional) independence to hold.
\begin{lemma} \label{lemma:cond-indep-basic}
For a submodular $f$, $A \perp_f \emptyset$. Also, $A \perp_f C \;|\;B$ if $A \subseteq B$ (and $C$ is any set) or $C \subseteq B$ (and $A$ is any set). Moreover, $A \perp_f C \;|\; B \implies I_f(A;C) \leq I_f(A;B), I_f(A;C) \leq I_f(C;B)$. 
\end{lemma}

\begin{proof}
We first prove that $A \perp_f \emptyset$. This follows from definition since $I_f(A; \emptyset) = f(A) + f(\emptyset) - f(A) = 0$ if $f$ is normalized. Next, if $A \subseteq B$, $A \cup B = B$ and hence $f(C | B) = f(C | A \cup B)$. Similarly, $f(A | B) = F(A | B \cup C) = 0$ and hence $I_f(A; B | C) = 0$. The same proof holds for thye case if $C \subseteq B$.

Finally, we show that $I_f(A; C) \leq I_f(A; B)$ (the case of $I_f(A; C) \leq I_f(C; B)$ follows from a symmetric argument). Recall that $I_f(A; C | B) = 0$ implies that $I_f(A; B \cup C) = I_f(A; B)$. This implies that $I_f(A; C) \leq I_f(A; B \cup C) = I_f(A; B)$. The first inequality follows from the monotonicity of $I_f$ in one argument given the other. Hence proved.
\end{proof}

The last result in Lemma~\ref{lemma:cond-indep-basic} is similar to the classical data processing inequality, but defined on sets of variables. Given that the sets $A \to B \to C$ form a Markov chain\footnote{This is in the usual sense of Markovicity, that knowing $B$, $A$ is fully determined irrespective of $C$ and so forth}, it implies that the conditional mutual information $I_f(A; C | B) = 0$. Furthermore, let $C = P(B)$ be a processing operator ($P: 2^{\Omega} \to 2^{\Omega}$). If the processing involves taking subsets $C \subseteq B$, this implies that $I_f(A; C | B) = 0$ (from Lemma~\ref{lemma:cond-indep-basic}). For specific sub-classes of functions, the processing can be more interesting. For example, in the case of set cover, let $\gamma^{-1}$ be the inverse operator such that given a concept $u \in U, \gamma^{-1}(u) \subseteq \Omega$
 gives the subset of $\Omega$ such that $\forall c \in \gamma^{-1}(u), u \in \gamma(c)$. Similarly, given a subset $C_U \subseteq U$, we can define $\gamma^{-1}(C_U)$. Then given sets $A, B \subseteq \Omega$, let $B_U \subseteq \gamma(B) \subseteq U$ be any subset. Then the sets $A \to B \to \gamma^{-1}(B_U)$ form a Markov chain with the property that $I_f(A; \gamma^{-1}(B_U) | B) = 0$. 
 
 Finally, we give some examples of properties which do not hold for combinatorial independence. 
 \begin{lemma}
 Suppose $A, B, C$ are subsets such that $A \perp_f B$ and $A \perp_f C$. This does not however, imply that $A \perp_f B \cup C$. 
 \end{lemma}
\begin{proof}
 To see this, again, define $f(A) = \min(|A|, k)$ and let $|A| = |B| = |C| = k/2$ with $A,B,C$ being mutually disjoint. Note that $A \perp_f B$ and $A \perp_f C$. Lets study $I_f(A; B \cup C) = f(A) + f(B \cup C) - f(A \cup B \cup C) = 3k/2 - k = k/2 > 0$. This means that $A$ is not jointly independent of $B \cup C$. This is a good segway to the next subsection, which studies independence among $k$-sets.
 \end{proof}

 \section{Multi-Set Submodular Independence}

 In this section, we first introduce multi-set submodular independence. In particular, we introduce two concepts of independence among $k$ sets. 
\begin{itemize}
    \item Sets $A_1, \cdots, A_k$ are mutually independent iff $f(\cup_i A_i) = \sum_i f(A_i)$.
    \item Sets $A_1, \cdots, A_k$ are pairwise independent iff $A_i \perp_J A_j, \forall i, j$.
\end{itemize}
 Next, we study the connection between mutual and pairwise independence of sets $A_1, \cdots, A_k$. We first define the multi-set total correlation $C_f$ as~\cite{iyer2021submodular}:
 \begin{align}
     C_f(A_1; \cdots; A_k) = \sum_{i = 1}^k f(A_i) - f(\cup_{i = 1}^k A_i)
 \end{align}
 Next, note that the mutual submodular independence between sets $A_1, \cdots, A_k$ means that $C_f(A_1; \cdots; A_k) = 0$. On the other hand, pairwise independendence implies that for all pairs, $A_i \perp_J A_j$. Both types of independence are again, w.l.o.g. defined on disjoint sets of items. The following result connects the two types of independences.
 \begin{lemma}
 Given a monotone, non-negative and normalized submodular function $f$, mutual independence implies pairwise independence. However, pairwise independence does not imply mutual independence.
 \end{lemma}
 \begin{proof}
 Lets prove that mutual independence implies pairwise independence. To prove this, lets assume that sets $A_1, \cdots, A_k$ are mutually independent, but not pairwise independent. That is, there exist two sets which are not joint independent. W.l.o.g, lets assume they are $A_1, A_2$. In other words, $f(A_1 \cup A_2) < f(A_1) + f(A_2)$. This implies (following from submodularity):
 \begin{align}
     f(\cup_i A_i) \leq f(A_1 \cup A_2) + \sum_{i = 3}^k f(\cup_i A_i)
 \end{align}
 Now invoking the assumption that $f(A_1 \cup A_2) < f(A_1) + f(A_2)$, this means that:
  \begin{align}
     f(\cup_i A_i) < f(A_1) + f(A_2) + \sum_{i = 3}^k f(\cup_i A_i)
 \end{align}
 which means that $f(\cup_i A_i) < \sum_{i = 1}^k f(\cup_i A_i)$ which contradicts mutual independence. Hence, given mutual independence, this must imply pairwise independence. However, pairwise independence does not imply mutual independence. We prove this with an example. Again, define $f(A) = \min(|A|, c)$ and let $|A_1| = |A_2| = \cdots = |A_k| = c/2$ be mutually disjoint sets. Note that all pairs $A_i, A_j$ are pairwise independent. However, they are not mutually independent since $f(\cup_i A_i) = c < \sum_i f(A_i) = kc/2$ for $k > 2$.
 \end{proof}
 Finally, we consider an alternative to mutual independence, which is $I_f(A_1; \cdots; A_k) = 0$ instead of $C_f(A_1; \cdots; A_k) = 0$. Recall that the multi-set mutual information $I_f$ is defined as:
 \begin{align}
 I_f(A_1; A_2;\dots; A_k) = -\sum_{T \subseteq [k]} (-1)^{|T|} f(\cup_{i \in T} A_i). 
\end{align}
 Note that $I_f$ and $C_f$ are \emph{different} quantities and \textbf{they will in general provide different conditions for the mutual independence between sets $A_1, \cdots, A_k$.} However, the condition that the multi-set submodular mutual information is zero is unfortunately not very interesting. To understand this, we look at it via the example of the set cover function.
 \begin{example}
 If $f(A) = \gamma(A)$, the multi-set submodular mutual information being zero is equivalent to $\cap_{i = 1}^k \gamma(A_i) = \emptyset$. This is, however, not very interesting since even if any two sets $A_1$ and $A_2$ are jointly independent, they will satisfy $\gamma(A_1) \cap \gamma(A_2) = \emptyset$ and hence it will hold that $\cap_{i = 1}^k \gamma(A_i) = \emptyset$. This means that even if only two sets are mutually independent, but the rest of the sets are completely dependent (or even identical), the multi-set submodular mutual information will still be zero.
 \end{example}
 
 \section{The Utility of Independence Characterizations}
Here, we present some discussion on the applications and utility of our formulations. In several data subset selection applications (e.g. video/image collection summarization~\cite{kaushal2020unified,kaushal2021prism,tschiatschekLearningMixturesSubmodular2014a}, data selection for efficient training~\cite{wei2015submodularity, killamsetty2020glister,killamsetty2021grad}, and active learning~\cite{ash2019deep,killamsetty2020glister,wei2015submodularity}, submodular functions have been shown to be a natural fit. In particular, many of these problems involve optimizing a submodular function under constraints, such as cardinality, knapsack, and matroid constraints~\cite{iyer2015submodular,tohidi2020submodularity}. 

Submodular Independence discussed in this work, can be viewed as a new class of combinatorial constraints, and we will refer to this constraint as $A \perp_f P$. In particular, we can then consider optimization problem as:
\begin{align}
    \max_{A \subseteq V} g(A), \mbox{s.t. } A \perp_f P
\end{align}
for a given set $P$. This has several natural applications. The first is privacy preserving summarization~\cite{kaushal2020unified,kaushal2021prism} where we want to select a subset $A$ which is as different as possible from set $P$ (and this difference is measured with respect to the submodular function). For instance this private set could be of one’s personal picture collection or medical data, or could be images of one’s family. The independence equates to lack of discernment of any information contained in $P$ by knowing $A$. If the independence considered here is JI, then this is equivalent to $I_f(A; P) = 0$, which we can equivalently relax to $I_f(A; P) \leq \epsilon$ for a very small $\epsilon$.  We can also consider the different types of independence here since the constraint $I_f(A; P) \leq \epsilon$ may not be amenable to tractable optimization algorithms when $I_f(A; P)$ is not submodular in A for a fixed $P$ (which is true for several important classes of submodular functions). In such cases, we can relax this to SMI or MI instead of JI, which requires that $I_f(a; P) = 0, \forall a \in A$ and $I_f(A; p) = 0. \forall p \in P.$ 

If $f$ is second-order supermodular~\cite{iyer2021submodular}, then the problem of maximizing $g(A)$ subject to a constraint that $I_f(A; P) \leq \epsilon$ is an instance of SCSK~\cite{iyer2013submodularScsk,iyer2013fast} which admits bounded approximation guarantees. For general $f$ though, achieving such a set could be NP hard.  On the other hand, the independence characterizations for MI (equivalently SMI) and PI are much easier. In particular, for MI, we just need to find elements $a \in A$ such that $f(a | B) = 0$. Similarly, for PI, we find elements $a \in A$ such that $f(a | b) = 0, \forall b \in B$. 

 \section{Proofs of Results in Section III}
Proof of Theorem~\ref{theorem:indep-relations}.
 \begin{proof}
We start with $JI \implies MI$. $JI$ implies that $f(A | B) = f(A)$. Let $A = \{a_1, \cdots, a_k\}$ be the elements in $A$, and define a chain of sets corresponding to the above ordering as $A = C_0 \subseteq C_1 \subseteq \cdots \subseteq C_k = A \cup B$. The set $C_i = A \cup \{a_1, \cdots, a_i\}$. Similarly, consider the chain of sets $\emptyset = D_0 \subseteq D_1 \subseteq \cdots \subseteq D_k = A$ with $D_i = \{a_1, \cdots, a_i\}$. Then note that $f(A | B) = \sum_{i =1}^{k} f(a_i | C_{i-1})$ and $f(A) = \sum_{i = 1}^k f(a_i | D_{i-1})$. Since $f(A | B) = f(A)$, this implies that $\sum_{i =1}^{k} f(a_i | C_{i-1}) = \sum_{i = 1}^k f(a_i | D_{i-1})$ . Note that $\sum_{i =1}^{k} f(a_i | C_{i-1}) \leq \sum_{i = 1}^k f(a_i | D_{i-1}), \forall i$ and hence this implies that for each $i$, $f(a_i | C_{i-1}) = f(a_i | D_{i-1})$. Setting $i = 1$, we have $f(a_i | B) = f(a_i)$. Note that since the order $a_1, \cdots, a_k$ can be arbitrary, we have $f(a | B) = f(a), \forall a \in A$.
 
Next, we show that $MI \iff SMI$. For this, first observe that  $MI \implies SMI$ by submodularity. From $MI$ we have, $f(\{a\}|B) = f(\{a\}) \; \forall a \in A$ and $f(\{b\}|A) = f(\{b\}) \; \forall b \in B$. Also with submodularity we have: $f(\{a\}|B) \leq f(\{a\}|X) \leq f(\{a\}) \; \forall X \subseteq B, a \in A$. But $f(\{a\}|B) = f(\{a\}) \implies f(\{a\}|X) = f(\{a\})$. Similar proof works for showing the other case with $B$.
For $SMI \implies MI$, we can just use $X = B$ (since we trivially have $B \subseteq B$). Similarly for the other case, use $X = A$.

We have $SMI \implies PI$ by just substituting in the singleton set as: $X = \{b\}, \; (b \in B)$ in $PI$ and it follows through trivially since $\{b\} \subseteq B \; \forall b \in B$. Similarly we can also substitute $X = \{a\}, \; (a \in A)$ for showing the other case.

For $ModI \implies JI$, recall that $ModI$ implies that $f(j | X) = f(j), \forall X \subseteq A \cup B$ and$j\notin A\cup B \backslash X$. Let $A = \{a_1, a_2, ..., a_k \}$ and $B = \{b_1, b_2, ..., b_k \}$. Furthermore, define $A_l = \{a_1, \cdots, a_l\}, l \leq k$ and $B_l = \{b_1, \cdots, b_l\}, l \leq k$. Then, $f(A) = \sum_{i = 1}^k f(a_i | A_{i-1}) = \sum_{i=1}^{k} f(a_i)=\sum_{j\in A} f(A)$. Similarly we obtain $f(B) = \sum_{j\in B} f(j), f(A \cup B) = \sum_{j \in A \cup B} f(i)$ and clearly, this satisfies $f(A) + f(B) = f(A \cup B)$ when $A$ and $B$ are disjoint. This shows that $ModI \implies JI \implies MI \iff SMI \implies PI$. 

We then show $SModI \implies MI$. $SModI$ implies $f(j | X) = f(j), \forall j \in A \cup B \backslash X$ and $X \subseteq A$ or $X \subseteq B$. Setting $X = A$ and $X = B$, we get $T2$. Also, $ModI \implies SModI$ since $SModI$ only requires $X \subseteq A$ and $ X \subseteq B$ while $ModI$ requires $X \subseteq A \cup B$. Finally, we show the reverse implications do not hold with examples. We start with $JI \centernot\implies ModI$. Let $f(A) = |\gamma(A)|$ and $A, B$ be sets such that $\gamma(A) \cap \gamma(B) = \emptyset$. Concretely, let $\Omega = \{1,2,3\}$, $\gamma(1) = \{c_1,c_2\}, \gamma(2) = \{c_1\}, \gamma(3) = \{c_3\}$, and $A = \{1,2\}, B = \{3\}$. Note that sets $A$ and $B$ are $JI$ independent since $f(A \cup B) = 3 = f(A) + f(B)$. However, it is not $ModI$ independent since let $X = \{1\}$, then $f(\{2\} | X) = 0 \neq f(\{2\}) = 1$. 

Next, we show that $SModI \centernot\implies ModI$. Let $f(A) = \min(|A|, k)$ with $A$ and $B$ being disjoint sets of size $k-1$. Then, $A$ and $B$ are $SModI$-independent since for all $X \subseteq A$ or $B$, $f(j | X) = f(j)$ (the function is modular if we only add a single element). However, they are not $ModI$ independent since we can set $X = A \cup \{b_1\}$ for a specific $\{b_1\} \in B$, and then selecting a $j \neq b_1 \in B$, we have that $f(j | X) = 0 \neq f(j)$. We then show that $JI \centernot\implies SModI$ and $SModI \centernot\implies JI$. To show $JI \centernot\implies SModI$, we use the set-cover function used earlier. The specific sets $A$ and $B$ are $JI$ independent however they are not $SModI$ independent since let $X = \{1\} \subseteq A$ and $j = \{2\}$. Then $f(\{2\} | \{1\}) = 0 \neq f(\{2\} = 1$. To show that $SModI \centernot\implies JI$, again use the matroid rank function $f(A) = \min(|A|, k)$ with $A$ and $B$ being disjoint sets of size $k-1$. Notice that $A$ and $B$ are $SModI$ independent. However, $f(A \cup B) = k \neq f(A) + f(B) = 2(k-2)$ and hence they are $JI$ independent.

Next, we show that $MI \centernot\implies SModI$ and $MI \centernot\implies JI$. $MI \centernot\implies SModI$ follows since $JI \implies MI$ and $JI \centernot\implies SModI$. Hence there exists a $f$, and sets $A, B$ which are $JI$ independent (and hence $MI$ independent) but not $SModI$ independent. Similarly, $MI \centernot\implies JI$ follows since $SModI \implies MI$ and $SModI \centernot\implies JI$. Finally, we show that $PI \centernot\implies MI$. Again, let $f(A) = \min\{|A|, 2\}$ with $|A| = |B| = 4$ and $A \cap B = \emptyset$. Note that every pair of elements $i \in A, j \in B$ satisfy $f(\{1,j\}) = 2 = f(\{i\} + f(\{j\})$. However, $f(j | A) = 0 \neq f(j) = 1, \forall j \in B$. Hence $A$ and $B$ are not $MI$ independent.
\end{proof}

Proof of Lemma~\ref{lemma:indep-rel-setcover}.
\begin{proof}
To prove this, we show that $JI \iff PI$. We know that $A \perp_f B$ (when $f(A) = w(\gamma(A))$) iff $\gamma(A) \cap \gamma(B) = \emptyset$. Similarly, $A \perp_f B$ with $PI$ if for every $a \in A, b \in B, \gamma(a) \cap \gamma(b) = \emptyset$. Note that this condition then also implies $\gamma(A) \cap \gamma(B) = \emptyset$ and hence $PI \implies JI$ which proves the equivalence. Since $JI \iff PI$, this means that all the four types are equivalent. Finally, $JI \centernot\implies ModI$ since the proof of Theorem~\ref{theorem:indep-relations} uses an instance of set cover to show that $JI \centernot\implies ModI$.
\end{proof}

Proof of Lemma~\ref{lemma:app-ent-ind-rel}.
\begin{proof}
Since entropy is a submodular function, the following holds from the proof of Lemma \ref{theorem:indep-relations}.
\begin{align}
   ModI \implies JI \implies MI \iff SMI \implies PI \; \\\text{and} \; ModI \implies SModI \implies MI
\end{align}

Next, we show that reverse implications do not hold. We start with $JI \centernot\implies ModI$. Consider $\Omega={X_1,X_2,X_3}$ with the joint distribution satisfying $P_{X_1,X_2,X_3}=P_{X_1,X_2}P_{X_3}$. For an example, let $X_1, X_2, X_3$ be three binary random variables (taking values in $\{ 0,1 \}$) such that $X_1, X_2$ are jointly distributed as $P_{X_1,X_2}(1,0) = P_{X_1,X_2}(0,1) = \frac{1}{8}, \; P_{X_1,X_2}(0,0) = \frac{1}{4}, \; P_{X_1,X_2}(1,1) = \frac{1}{2}$, and we have $X_3 \sim \text{Bern}(\frac{1}{2})$. 
Take $A=\{1,2\},\ B=\{3\},\ X=\{2\}, j=\{1\}$. Since the pair $X_1,X_2$ are independent of $X_3$, we have $H(A|B)=H(A)$ and $H(B|A)=H(A)$, i.e., they are $JI$ independent but we have $H(j|X) \neq H(j)$, implying A and B are not $ModI$ independent. 

For showing $JI\centernot\implies SModI$ we consider the same example as above with $P_{X_1,X_2,X_3}=P_{X_1,X_2}P_{X_3}$ and let $A=\{1,2\},\ B=\{3\},\ X=\{2\}, j=\{1\}$. $A$ and $B$ are $JI$ independent as shown above. However, this would not imply necessarily $SModI$ independence criterion which needs here pairwise independence between all the variables and that $H(X_3|X_1,X_2)=H(X_2)$.

Next, we show $SModI\centernot\implies ModI$ and $SModI\centernot\implies JI$. To prove them both, we consider the example with $\Omega=\{1,2,3,4\}$ with $X_i\in\{0,1\}$, $\forall\ i=1,2,3,4$ and $P_{X_1,X_2,X_3,X_4}(x_1,x_2,x_3,x_4)=\frac{1}{8}$ for $(x_1,x_2,x_3,x_4)\in \{(0,0,0,0),(0,0,1,1),(0,1,0,1),(0,1,1,0),(1,0,0,1),\\ (1,0,1,0),(1,1,0,0),(1,1,1,1)\}$. Note here we have $X_4=X_1\oplus X_2\oplus X_3\oplus X_4$. Let us have $A=\{1,2\}$ and $B=\{3,4\}$. Thus, with the chosen joint distribution we have $SModI$ to be true which needs pairwise independence and the mutual independence between any set of three random variables. Since $H(X_4|X_1,X_2,X_3)=0\neq H(X_4)=\frac{1}{2}$, $ModI$ is not true. Also as $P_{X_1,X_2,X_3,X_4}(0,0,0,0)=\frac{1}{8}\neq P_{X_1,X_2}(0,0)P_{X_3,X_4}(0,0)=\frac{1}{16}$, $ModI$ is not true. Therefore, $SModI\centernot\implies ModI$ and $SModI\centernot\implies JI$.

$MI \centernot\implies SModI$ follows since $JI \implies MI$ and $JI \centernot\implies SModI$. Similarly, $MI \centernot\implies JI$ follows since $SModI \implies MI$ and $SModI \centernot\implies JI$.

In the end, we show $PI \centernot\implies MI$. Here consider an example, where the assumed probability distribution on $X_1,X_2,X_3$ which take values in $\{0,1\}$, is as follows: $P_{X_1,X_2,X_3}(0,0,0)=P_{X_1,X_2,X_3}(0,1,1)=P_{X_1,X_2,X_3}(1,0,1)=P_{X_1,X_2,X_3}(1,1,0)=\frac{1}{4}$. It is easy to check here that there is pairwise independence, as, $X_i \sim \text{Bern}(\frac{1}{2})$ and that  $P_{X_i,X_j}(0,0) = P_{X_i,X_j}(1,0) = P_{X_i,X_j}(0,1) = P_{X_i,X_j}(1,1)=\frac{1}{4}$, $\forall i\neq j\in\{1,2,3\}$. 
Thus there is pairwise independence which implies $PI$ independence for sets $A=\{1,2\}$ and $B=\{3\}$. However there is no mutual independence here as we can see that $P_{X_1,X_2,X_3}(0,0,0)=\frac{1}{4}\neq\frac{1}{8}=P_{X_1}(0)P_{X_2}(0)P_{X_3}(0)$. In fact, by the way of construction of the example, knowing $X_1$ and $X_2$ one knows $X_3$ with certainty, which implies $H(X_3|X_1,X_2)=0 \neq H(X_3)=1$. This implies $MI$ independence does not hold for the sets as it requires at least $H(X_3|X_1,X_2)=H(X_3)$, which is not true in this example.
\end{proof}

\section{Conclusions}
To conclude, in this paper, we study different independence classes of combinatorial independence, and their relationships. We then discuss the implications of these results in the entropic case, and provide algorithms for obtaining the independent sets for different independence classes. Finally, we discuss some implications of the results for applications like summarization and clustering.
\bibliographystyle{plain}
\bibliography{references}

\IEEEtriggeratref{3}
\ifCLASSINFOpdf
\else
\fi

\end{document}